\documentclass[sn-mathphys]{sn-jnl}

\jyear{2021}%
\usepackage{fancyhdr}
\usepackage{mathrsfs}
\usepackage[mathscr]{euscript}
\usepackage{graphicx}
\usepackage{amssymb}
\usepackage{comment}
\usepackage{amsmath}

\usepackage{amsthm}
\usepackage{enumitem}
\usepackage{breakurl}
\usepackage[english]{babel}
\usepackage{hyperref}
\hypersetup{colorlinks=true}
\usepackage{cleveref}

\newcommand{\F}{\mathscr F}

\def\e{{\varepsilon}}

\newcommand{\basePoint}{\mathcal{B}}
\newcommand{\Int}{\mathbb{Z}}
\newcommand{\G}{\mathbb{G}}

\newcommand{\adv}{ \mathscr{A}}
\newcommand{\V}{ \mathcal{V}}

\newcommand{\U}{ \mathcal{U}}
\newcommand{\E}{ \mathbb{E}}
\newcommand{\ab}{ \mathcal{A}}
\newcommand{\rb}{ \mathcal{R}}
\newcommand{\mb}{ \mathcal{M}}
\newcommand{\cb}{ \mathcal{C}}
\newcommand{\ub}{ \mathcal{U}}
\newcommand{\kb}{ \mathcal{K}}

\newcommand{\sh}{ y}
\newcommand{\shb}{ \mathcal{Y}}
\newcommand{\rtb}{ \mathcal{R}'}

\newcommand{\Pb}{ \mathbb{P}}

\newcommand{\pk}{\mathsf{pk}}
\newcommand{\sk}{\mathsf{sk}}
\newcommand{\rec}{\mathsf{rec}}

\DeclareMathOperator{\Com}{Com}
\DeclareMathOperator{\Ver}{Ver}

\theoremstyle{thmstyleone}
\newtheorem{thm}{Theorem}[section]
\newtheorem{lem}{Lemma}[section]

\theoremstyle{thmstylethree} 
\newtheorem{defn}{Definition}[section]
\newtheorem{oss}{Observation}

\begin{document}


\title{A Provably-Unforgeable Threshold EdDSA with an Offline Recovery Party}

\author*[1]{\fnm{Michele} \sur{Battagliola}}\email{battagliola.michele@gmail.com}
\equalcont{These authors contributed equally to this work.}

\author[1]{\fnm{Riccardo} \sur{Longo}}\email{riccardolongomath@gmail.com}
\equalcont{These authors contributed equally to this work.}
\author[1]{\fnm{Alessio} \sur{Meneghetti}}\email{alessio.meneghetti@unitn.it}
\equalcont{These authors contributed equally to this work.}
\author[1]{\fnm{Massimiliano} \sur{Sala}}\email{maxsalacodes@gmail.com}
\equalcont{These authors contributed equally to this work.}

\affil*[1]{\orgdiv{Department of Mathematics}, \orgname{University Of Trento}, \orgaddress{\street{Via Sommarive 14}, \city{Povo}, \postcode{38123}, \state{TN}, \country{Italy}}}


\abstract{   We present an EdDSA-compatible multi-party digital signature scheme that supports an offline participant during the key-generation phase, without relying on a trusted third party.
    \\
    Under standard assumptions we prove our scheme secure against adaptive malicious adversaries.
    Furthermore, we show how our security notion can be strengthen when considering a rushing adversary.
    Moreover we provide a possible solution to achieve the resiliency of the recovery in the presence of a malicious party.
    Using a classical game-based argument, we prove that if there is an adversary capable of forging the scheme with non-negligible probability, then we can build a forger for the centralized EdDSA scheme with non-negligible probability.}

\keywords{94A60 Cryptography, 12E20 Finite fields, 14H52 Elliptic curves, 94A62 Authentication and secret sharing, 68W40 Analysis of algorithms}



\maketitle

\section{Introduction}

A $(t,n)$-threshold signature scheme is a multi-party computation protocol that enables a subset of at least $t$ among $n$ authorized players to jointly perform digital signatures. 
The flexibility and security advantages of threshold protocols have become of central importance in the research for new cryptographic primitives \cite{nist_stand}. 
\\
Starting from the highly influential work of Gennaro et al~\cite{gennaro1996robust}, several authors proposed both novel schemes \cite{mackenzie2001two,lindell2018fast,canetti2020uc} and improvements to existing protocols \cite{ mackenzie2004two,gennaro2016threshold,boneh2017using,lindell2017fast,doerner2018secure,gennaro2018fast,doerner2019threshold,kondi2019refresh}. 

Usually threshold schemes translate in the multi-party setting a well-established signature scheme, while producing signatures that are compatible with this centralized version.
Then their security is proved with a reduction to the standard centralized scheme, like the proof presented in \cite{ECDSA2020threshold} and \cite{gennaro}.
The key-generation and signature algorithms are replaced by a communication protocol between the parties, while the compatibility is achieved by keeping the verification algorithm of the centralized algorithm.
This approach streamlines the insertion of the new protocol in the cryptographic landscape, because verification is compatible with established solutions and the resulting security derives from standard assumptions.

Recently, in \cite{ECDSA2020threshold} the authors propose a $(2,3)$-threshold multi-signature protocol compatible with ECDSA in which one of the users plays the role of recovery party: a user involved only once in a preliminary set-up prior even to the key-generation step of the protocol. 
More precisely, only two parties are active in the key-generation step of the protocol, and by a secure multi-party protocol they create their own private keys together with some additional data to be eventually sent to the third non-active player.
In case of need, the third player receives and uses this additional data to generate its private key and can therefore eventually participate in the signature phase with one of the other two.
Threshold multi-signatures with offline parties are applied for example in custody services for crypto-assets \cite{dinicola_2020}.

In this paper we propose an EdDSA-compatible variant of \cite{ECDSA2020threshold} where again the key-generation algorithm of the protocol does not require the active involvement of all three players.
\\
While EdDSA offers better performance than ECDSA, the latter is at first glance better suited for a multiparty environment: the presence of hash computations in EdDSA is indeed not readily-compatible with an MPC setting.
In order to work around the problem, we work with a variant of EdDSA whose ouputs are indistinguishable from those of the standard version, and we adopt some techniques similar to those in \cite{Purify} to  deal with the deterministic nature of the protocol.

We prove the protocol secure against adaptive adversaries by reducing it to the classical EdDSA scheme, assuming the security of a non-malleable commitment scheme, the strength of the underlying hash function and an IND-CPA encryption scheme.
Moreover we make some considerations about the resiliency of the recovery, an interesting aspect due to the presence of an offline party, analyzing possible changes that allow us to achieve this higher level of security.

\paragraph{Organization}
We present some preliminaries in \Cref{preliminaries}. We describe our protocol in \Cref{protocol}, in particular in \Cref{key-derivation} we provide a protocol extension that includes key-derivation.
In \Cref{security} we state and prove the security properties of our protocol.
Finally in \Cref{conclusions} we draw our conclusions.

\section{Preliminaries}\label{preliminaries}
In this section we present some preliminary definitions and primitives that will be used in the protocol and its proof of security.

\paragraph{Notation}
We use the symbol $\parallel$ to indicate the concatenation of bit-strings.
Sometimes we slightly abuse the notation and concatenate a bit-string $M$ with an elliptic curve point $\mathcal{P}$, in those cases we assume that there has been fixed an encoding $\varphi$ that maps elliptic curve points into bit-strings, so $M\parallel\mathcal{P} := M\parallel\varphi(\mathcal{P})$.

In the following when we say that an algorithm is \emph{efficient} we mean that it runs in (expected) polynomial time in the size of the input, possibly using a random source.

We use a blackboard-bold font to indicate algebraic structure (i.e. sets, groups, rings, fields and elliptic curves), and a calligraphic font to denote points over elliptic curves.
About elliptic curves, we distinguish the notation for the curve used for the signature and for the auxiliary curve used in the deterministic nonce generation: the latter are characterized by a prime symbol.

\subsection{Decisional Diffie-Hellman Assumption}
Our proof is based on the Decisional Diffie-Hellman~\cite{boneh1998decision} (from now on DDH).

\begin{defn}[DDH Assumption]\label{DDH}
Let $\G$ be an (additive) cyclic group with generator $\basePoint$ and order $q$.
Let $a, b, c$ be random elements of $\Int_q$.
The Decisional Diffie-Hellman (DDH) assumption, states that no efficient algorithm can distinguish between the two distributions $(\basePoint, a\basePoint, b\basePoint, ab\basePoint)$ and $(\basePoint, a\basePoint, b\basePoint, c\basePoint)$.
\end{defn}

\subsection{Cryptographic Hash Functions}\label{hash}
In the EdDSA scheme (and therefore in our threshold protocol) a cryptographic hash function $H$ is used as a \emph{Pseudo-Random Number Generator} (PRNG), employed to derive secret scalars and nonces.

For this reason we map the output of the hash function onto the ring $\Int_q$ where $q$ is a prime and the order of the base point $\basePoint$ used in EdDSA (i.e. $\basePoint$ generates a subgroup of elliptic curve points with prime order $q$), and we require $H$ to behave like a \emph{Random Oracle}.
We formalise our requirements with the following definition.
\begin{defn}[Good PRNG] \label{good-PRNG}
  Let $H: \{0,1\}^* \rightarrow \Int_q$ be a function that maps bit-strings of arbitrary length into elements of $\Int_q$.
  $H$ is a \emph{Good PRNG} if no efficient algorithm can distinguish between the distributions of $H(S)$ and $x$, where both ${x \in \Int_q}$ is chosen uniformly at random, and $S \in \{0,1\}^*$ is a bit-string that embeds at least $n$ bits of entropy, with $2^{n} < q < 2^{n+1}$.
\end{defn}
This definition is not standard, but precisely captures exactly what we need from a hash function to generate a good nonce.
Also note that the stronger classical definition of a Random Oracle, that is usually used to study the security of EdDSA, perfectly satisfies our definition.

For secret scalars, EdDSA uses the hash function in a slightly more complicated way, in order to to prevent timing leaks in poor implementations, put a lower bound on standard attacks, and embed the curve cofactor into the scalar, so that even a multiplication by an adversary-controlled point would not leak information about the secret (although note that this does not happen in the EdDSA scheme).
For this reason we introduce also the following definition that captures this additional security requirement.
\begin{defn}[Strong PRNG] \label{strong-PRNG}
  Let $H: \{0,1\}^b \rightarrow \{0,1\}^{n}$ be a function that maps bit-strings of length $b$ into bit-strings of length $c \le n \le b$, with $q < 2^{n+2-c}$, and $c\in\{2,3\}$.
  $H$ is a \emph{Strong PRNG} if no efficient algorithm can distinguish between the distributions of $\psi(H(\mathtt{k}))$ and $x$, where both $\mathtt{k}\in \{0,1\}^b, x \in \Int_q$ are chosen uniformly at random, and $\psi: \{0,1\}^{n}\rightarrow \Int_q$ is defined as:
  \begin{equation}\label{scalar-encoding}
      \psi(h) = 2^{n+1}+\sum_{i=c}^{n}2^i h_i \mod q.
  \end{equation}
\end{defn}
Again, this definition is not standard, but is a tight fit for what we need and a classical Random Oracle satisfies it.

\subsection{EdDSA} \label{EdDSA}
Edwards-curve Digital Signature Algorithm (EdDSA)~\cite{eddsa} is a digital signature scheme based on twisted Edwards curves. It is designed to be faster than the previously developed schemes without sacrificing security.

EdDSA has several parameters: a prime field $\mathbb{F}_p$; an integer $b$ with $2^{b-1} > p$; a $(b-1)$-bit encoding of elements of the finite field $\mathbb{F}_p$ (if omitted it is assumed to be the classical \emph{little-endian} encoding); a cryptographic hash function $H$ producing $2b$-bit outputs; an integer $c \in\{2, 3\}$ associated to the cofactor of the curve, an integer $n$ with $c \le n \le b$ (secret scalars are $n+1$ bits long); a non-zero square element $a\in\mathbb{F}_p$; a non-square element $d$ of  $\mathbb{F}_p$; a point $\basePoint \ne (0,1)$ of the curve described by the equation:
\begin{equation}
    a x^2+y^2 = 1+d x^2y^2,
\end{equation}
and a prime $q$ such that $q\basePoint=0$ and $2^c q$ is the number of points of the curve.
Elliptic curve points are encoded as $b$-bit strings that are the $(b-1)$-bit encoding of their second coordinate $y$, followed by a \emph{sign bit} that is set if the $(b-1)$-bit encoding of the first coordinate $x$ is lexicographically larger than the
$(b-1)$-bit encoding of $-x$.
When we concatenate a point and a bit-string (e.g. $\mathcal{P}\parallel S$) we implicitly encode the point into a bit-string as explained above.
\\
\\
Given the parameters $(p, b, H, c, n, a, d, \basePoint, q)$ described above, the protocol works as follows:
\begin{enumerate}
	\item Choose a random $b-bit$ string $\mathtt{k}$, that will be the secret key.
	\item Compute $H(\mathtt{k}) = (h_0,...,h_{2b-1})$.
	\item Compute $a=\psi(h_0\parallel\ldots\parallel h_{n-1})$ (where $\psi$ is the same as~\Cref{scalar-encoding}), the public key is set to be $\mathcal{A}=a\basePoint$.
	\item To sign a message $M$ compute $r=H(h_b\parallel...\parallel h_{2b-1}\parallel M)$ (interpreting the digest as an integer), and  ${\mathcal{R}=r\basePoint}$.
	\item The signature is $(\mathcal{R},S)$, where $S=(r+a H(\mathcal{R}\parallel\mathcal{A}\parallel M)) \mod l$.
	\item to verify the signature check if $2^c S\basePoint = 2^c \mathcal{R}+2^c H(\mathcal{R}\parallel\mathcal{A}\parallel M)\mathcal{A}$.
\end{enumerate}

\subsection{Encryption Scheme}\label{IND-CPA}
In our protocol we need an asymmetric encryption scheme to communicate with the offline party.  
The minimum requirement we ask for our protocol to be secure is that the encryption scheme chosen by the offline party has the property of IND-CPA \cite{bellare2005introduction,indcpa}, i.e.:

\begin{defn}
    Let $\Pi = (\mathsf{Gen}, \mathsf{Enc}, \mathsf{Dec})$ be a public key encryption scheme. Let us define the
    following experiment between an adversary $\adv$ and a challenger $\mathscr{C}^b$ parametrized by a bit $b$:
    \begin{enumerate}
        \item The challenger runs $\mathsf{Gen}(1^k)$ to get $\sk$ and $\pk$, the secret and public keys. Then it gives $\pk$ to $\adv$.
        \item $\adv$ outputs two messages $(m_0, m_1)$ of the same length.
        \item The challenger computes $\mathsf{Enc}(\pk, m_b)$ and gives it to $\adv$.
        \item $\adv$ outputs a bit $b'$(if it aborts without giving any output, we just set $b'=0$). The challenger returns $b'$ as the output of the game.
    \end{enumerate}
    We say that $\Pi$ is secure against a chosen plaintext attack if for any k and any probabilistic polynomial time adversary $\adv$ the function
    \begin{equation}
        \mathsf{Adv}(\adv) = \Pb[\mathscr{C}^1(\adv, k) = 1] - \Pb[\mathscr{C}^0(\adv, k)= 1],
    \end{equation}
    i.e. $\mathsf{Adv}(\adv) = \Pb[b'=b] - \Pb[b'\ne b]$, is negligible.
 \end{defn}
 
This hypothesis will be enough to prove the unforgeability of the protocol, but it is possible to achieve an higher notion of security using more sophisticated encryption scheme that supports ZKP for the Discrete Logarithm. This will be more clearly explained in \Cref{SottosezioneResilienzaRecovery}.

\subsection{Commitment Schemes}\label{commitments}
A commitment scheme~\cite{brassard1988minimum} is composed by two algorithms:
\begin{itemize}
 \item \textbf{$\Com(M): \{0,1\}^* \rightarrow \{0,1\}^* \times \{0,1\}^*$:} takes in input the value $M$ to commit\footnote{In the protocol and the simulations we implicitly encode every value we need to commit into a bit-string, assuming there is a standard encoding understood by all parties} and, using a random source, outputs the commitment string $C$ and the decommitment string $D$.
 \item \textbf{$\Ver(C, D):  \{0,1\}^* \times \{0,1\}^* \rightarrow\{0,1\}^*\cup \{\perp\}$:} takes the commitment and decommitment strings $C, D$ and outputs the originally committed value $M$ if the input pair is valid, $\perp$ otherwise\footnote{Again, in the protocol we implicitly decode valid decommitment outputs (i.e. $\ne \perp$) into the original value, assuming that the decoding is also standard and understood by all parties}.
\end{itemize}
We require a commitment scheme to have the following properties:
\begin{itemize}
    \item \textbf{Correctness:} for every value $M$ it holds $\Ver(\Com(M)) = M$.
    \item \textbf{Binding:} for every commitment string $C$ it is infeasible to find ${M\ne M'}$ and $D\ne D'$ such that $\Ver(C,D) = M$ and $\Ver(C, D') = M'$ with both $M, M\ne \perp$. 
    \item \textbf{Hiding:} Let $(C, D) = \Com(M_b)$ with $b\in \{0, 1\}$, $M_1 \ne M_2$, then it is infeasible for an attacker that may choose $M_0 \ne M_1$ and sees only $C$, to correctly guess $b$ with more than negligible advantage.
    \item \textbf{Non Malleability:} Given $C = \Com(M)$, it is infeasible  for an adversary $\adv$ to produce another commitment string $C'$ such that after seeing $D$ such that $\Ver(C, D) = M$, $\adv$ can find a decommit string $D'$ such that $\Ver(C', D') = M'$ with $M'$ \emph{related} to $M$, that is $\adv$ can only create commitments to values that are independent from $M$.
\end{itemize}

\subsection{Zero-Knowledge Proofs}\label{zkp}
In the protocol various Zero-Knowledge Proofs (ZKP)~\cite{goldreich1986proofs} are used to enforce the respect of the passages prescribed by the specifications.
In fact in the proof of security we can exploit the soundness of these sub-protocols to extract valuable information from the adversary, and their zero-knowledge property to simulate correct executions even without knowing some secrets.
We can do so because we see the adversary as a (black-box) algorithm that we can call on arbitrary input, and crucially we have the faculty of rewind its execution.

In particular we use ZKP \emph{of Knowledge} (ZKPoK) to guarantee the usage of secret values that properly correspond to the public counterpart, specifically the Schnorr protocol for discrete logarithms, and its variant that proves that two public values are linked to the same secret (see~\cite{schnorr1989efficient,shoup2007sigma} and~\Cref{Schnorr}).
The soundness property of a ZKPoK guarantees that the adversary must know the secret input, and opportune rewinds and manipulations of the adversary's execution during the proof allows us to extract those secrets and use them in the simulation.
Conversely exploiting the zero-knowledge property we can trick the adversary in believing that we know our secrets even if we don't, thus we still obtain a correct simulation of our protocol form the adversary's point of view.

However Schnorr's protocol requires a prime order group, so we implicitly use the Ristretto technique~\cite{ristretto} for constructing prime order elliptic curve groups, and we transform elliptic curve points in Ristretto points for these computations.
This method extends Mike Hamburg's Decaf~\cite{hamburg2015decaf} approach to cofactor elimination to support cofactor-8 curves such as Curve25519~\cite{josefsson2017edwards, bernstein2006curve25519} (the standard EdDSA curve).
We refer to the original sources for more details about this approach.

\subsection{Feldman-VSS}\label{feldman}
Feldman's VSS scheme\cite{paperFeldman} is a verifiable secret sharing scheme built on top of Shamir's scheme\cite{shamirSS}. A secret sharing scheme is verifiable if auxiliary information is included, that allows players to verify the consistency of their shares.
We use a simplified version of Feldman's protocol: if the verification fails the protocol does not attempt to recover excluding malicious participants, instead it aborts altogether.
In a sense we consider \emph{somewhat honest} participants, for this reason we do not need stronger schemes such as~\cite{gennaro1999secure,schoenmakers1999simple}.
\\The scheme works as follows:
\begin{enumerate}
 \item A cyclic group $\G$ of prime order $q$ is chosen, as well as a generator $\basePoint \in \G$. 
 The group $\G$ must be chosen such that the discrete logarithm is hard to compute.
 \item The dealer computes a random polynomial $P$ of degree $t$ with coefficients in $\Int_q$, such that $P(0)=s$ where $s\in\Int_q$ is the secret to be shared.
 \item Each of the $n$ share holders receive a value $P(i)\in\Int_q$.
 So far, this is exactly Shamir's scheme. 
 \item To make these shares verifiable, the dealer distributes commitments to the coefficients of $P$. Let $P(X)=s+\sum_{i=1}^n a_i X^i$, then the commitments are $\cb_0=s\basePoint$ and $\cb_i={a_i}\basePoint$ for $i\in\{1,\ldots,n\}$.
 \item Any party can verify its share in the following way: let $\alpha$ be the share received by the $i$-th party, then it can check if $\alpha=P(i)$ by verifying if the following equality holds:
 $$\alpha\basePoint = \sum_{j=0}^t {(i^j)}\cb_j = s\basePoint +\sum_{j=1}^t {a_j(i^j)}\basePoint = \left(s + \sum_{j=1}^t a_j(i^j)\right)\basePoint={P(i)}\basePoint.$$
\end{enumerate}
In the proof we will need to simulate a $(2,2)$-threshold instance of this protocol without knowing the secret value $s$.

Let us use an additive group with generator $\basePoint$, and let $\shb = s \basePoint$, the simulation proceeds as follows:
\begin{itemize}
 \item the dealer selects two random values $a,b$ and forces $P(1) = a$, $P(2) = b$;
 \item\label{ECVSS} then sets $\cb_0=\shb$ and, depending on whether the adversary is $P_1$ or $P_2$, it computes:
 	\begin{align}
	 \cb_1&=a\basePoint-\shb;
	\end{align}
	in the case the adversary is $P_1$, or
	\begin{align}
	 \cb_1&=\frac{1}{2}(b\basePoint-\shb);
	\end{align}
	in the case the adversary is $P_2$.
	\item In either case the other player can successfully verify their shards, performing the corresponding check:
	\begin{align}
	 a\basePoint &= \shb + \cb_1 = \shb +a\basePoint-\shb
	 \end{align}
	 or 
	 \begin{align}
	 b\basePoint &= \shb + 2 \cb_1 = \shb+2\cdot\frac{1}{2}(b\basePoint-\shb).
	\end{align}
\end{itemize}

\subsection{Deterministic nonce generation}\label{ZKPpurify}

One of the peculiar features of EdDSA is that it is a deterministic signature algorithm, in the sense that it does not require the generation of a random nonce.

To achieve the same feature we rely on a verifiable random nonce generator: roughly each player chooses a random seed during the Key Generation algorithm and each time a signature is produced it is proven that the nonce used in the signature algorithm is coherent with the seed.

In particular we use Purify \cite{Purify}, a Pseudo-Random Function (PRF) purely based on elliptic curves. 

Let $\E$ be an Edward curve for the EdDSA algorithm, with a point $\basePoint$ of order $q$, Purify requires the choice of a second elliptic curve $\E'$ over $\mathbb{F}_{q^2}$ whose group of points is cyclic of order $q'$ is generated by a point $\basePoint'$ and is such that the DDH assumption holds.
In particular the participants fix a quadratic non residue $\delta \in \mathbb{F}_q^*$ and find $a,b \in \mathbb{F}_q$ such that
\begin{itemize}
    \item the equation $y^2 = x^3 + ax + b$ defines an elliptic curve ${\E}_1$ over $\mathbb{F}_q$ of a prime order $q_1$ in which the DDH assumption holds;
    \item the equation $y^2 = x^3 + a\delta^2x + b\delta^3$ defines an elliptic curve ${\E}_1$ over $\mathbb{F}_q$ of a prime order $q_2\ne q_1$ in which the DDH assumption holds;
\end{itemize}
Then define $\E'$ as the elliptic curve defined by the equation $y^2 = x^3 + ax + b$ over $\mathbb{F}_{q^2}$.
It is possible to prove that there is an efficiently computable and invertible isomorphism $\phi : \E' \to {\E}_1 \times {\E}_2$.

Let $z \in \{0,1\}^*$ be a string, we define the hash function
\begin{equation}\label{EqPur}
    H_{\mathtt{Pur}}(z) = \phi^{-1}(H_1(z),H_2(z))
\end{equation}
where $H_1$ and $H_2$ are hash functions onto $\E_1$ and $\E_2$ respectively.

Now let $f:\E' \to \Int_q$ be the function defined as follows:
\begin{equation}\label{fDef}
    f(\mathcal{Q}) =
    \begin{cases}
        0 & \text{if }\mathcal{Q}=0_{\E'}\\
        x_0 & \text{if }\mathcal{Q}=(x_0+x_1\sqrt{\delta},y_0+y_1\sqrt{\delta})\\
    \end{cases}
\end{equation}
 It is possible to prove that the uniform distribution over $\Int_q$ is statistically close to $f(U_{\E'} )$, where $U_{\E'}$ is the uniform distribution over $\E'$.
 So, if $z$ is a random uniformly distributed string, $u$ is distributed uniformly in $\mathbb{Z}_{q'}$, and $H_1, H_2$ behave like random oracles, we have that $f(u H_\mathtt{Pur}(z))$ is uniformly distributed in $\Int_q$.

The crucial aspect of this construction is that it allows the possibility of building a non-interactive ZKP for the relation $\U'=u'\basePoint'$ and $\rb=f(u\V')\basePoint$ where $\basePoint, \basePoint'$ are defined as before, $\V'$ is a public random point of $\E'$ and $u'$ is the private input of the prover.
This allows the construction of a verifiable pseudo-random nonce generator.
In particular the ZKP is described in \cite{Purify} and makes use of the Bulletproof framework \cite{Bulletproof}.

Formally the security of Purify is stated in the following Lemma:

\begin{lem}\label{LemmaPurify}
Let $\E'$ be an elliptic curve over $\mathbb{F}_{q^2}$ whose group of points generated by a point $\basePoint'$ is cyclic of order $q'$ and is such that the DDH assumption holds.
Let $u$ be a random element of $\Int_{q'}$ and $H_{\mathsf{Pur}}$ and $f$ be defined as respectively in \Cref{EqPur} and \Cref{fDef}.
Then $H_{\mathsf{Pur}}$  is indistinguishable from a random oracle onto $\mathbb{E}'$ and thus also $f(u H_{\mathsf{Pur}}(\cdot))$ is indistinguishable from a random oracle onto $\Int_q$.

Moreover it is possible to build a secure non-interactive ZKP for the relation relation $\U'=u'\basePoint'$ and $\rb=f(u\V')\basePoint$ where $\basePoint, \basePoint'$ are public data defined as before, $\V'$ is a public random point of $\E'$ and $u'$ is the private input of the prover.
\end{lem}

\begin{oss}
As it will be clear later, in our construction the signature is deterministic as long as the set of signers  is fixed.
    To achieve a deterministic signature that depends only on the message, an alternative solution is the usage of a multi-party symmetric cipher with authenticated MAC key such as MiMC\cite{MiMC} and the Marvellous\cite{Marvellous} family combined with a threshold secret sharing of the key.
    However, while being suited for multi-party nonce generation, these protocols have the drawback of requiring expensive precomputation steps, that are cumbersome in our settings.
\end{oss}

\section{Protocol Description}\label{protocol}
In this section we describe the details of our protocol.
After some common parameters are established, one player chooses a long-term asymmetric key and then can go offline, leaving the proper generation of the signing key to the remaining two participants. 
For this reason the signature algorithm is presented in two variants, one used jointly by the two players who performed the Key Generation, and one used by the offline player and one of the others.
\\More specifically the protocol is comprised by four phases:
\begin{enumerate}
 \item \textbf{Setup Phase} (\Cref{setup}): played by all the parties, it is used to decide common parameters. Note that in many contexts these parameters are mandated by the application, so the parties merely acknowledge them, possibly checking they respect the required security level.
 \item \textbf{Key Generation} (\Cref{key-gen}): played by only two parties, from now on $P_1$ and $P_2$. It is used to create a public key and the private shards for each player.
 \item \textbf{Ordinary Signature} (\Cref{ordinarysignature}): played by $P_1$ and $P_2$. As the name suggests this is the normal use-case of the protocol.
 \item \textbf{Recovery Signature} (\Cref{recoverysignature}): played by $P_3$ and one between $P_1$ and $P_2$. This models the unavailability of one player, with $P_3$ stepping up as a replacement.
\end{enumerate}

In order to obtain a deterministic signature scheme we need to rely on a verifiable nonces generation algorithm. We choose to use Purify, described in \cite{Purify}. Starting from some secrete parameters and their commited value, this algorithm allow every party to check whether the other party has computed the correct random value or not. 

From here on with the notation ``$P_i$ does something'', we mean that both $P_1$ and $P_2$ perform the prescribed task independently.
Similarly, the notation ``$P_i$ sends something to $P_j$'' means that $P_1$ sends to $P_2$ and $P_2$ sends to $P_1$.

\subsection{Setup Phase}\label{setup}
 This phase involves all the participants and is used to decide the parameters of the algorithm.
 \\The parameters involved are the following:
 \\
 \\
 \begin{tabular}{|ll|}
 	\hline
 	\textbf{Player 1 and 2} & \\ \hline
 	\textbf{Input:} & $-$ \\
 	\textbf{Private Output:} & $-$\\
 	\textbf{Public Output:} & $\E, \basePoint, q, H$\\
 	&$\E', \basePoint', q',$\\
 	\hline
 \end{tabular}
 \begin{tabular}{|ll|}
 	\hline
 	\textbf{Player 3} & \\ \hline
 	\textbf{Input:} & $-$ \\
 	\textbf{Private Output:} & $\sk_3$\\
 	\textbf{Public Output:} & $\pk_3$\\
 	&\\
 	\hline
 \end{tabular}
 \\
 \\
 $P_3$ chooses an asymmetric encryption algorithm and a key pair $(\pk_3, \sk_3)$, then it publishes $\pk_3$, keeping $\sk_3$ secret. $\pk_3$ is the key that $P_1$ and $P_2$ will use to communicate with $P_3$. 
 The algorithm which generates the key pair ($\sk_3$, $\pk_3$) and the encryption algorithm itself are unrelated to the signature algorithm, but it is important that both of them are secure.
 We require the encryption protocol to be IND-CPA, see~\Cref{IND-CPA} and~\Cref{SottosezioneResilienzaRecovery} for more details.
 
 Then $P_1$ and $P_3$ need to agree on a secure hash function $H$ whose outputs we interpret as elements of $\Int_q$, a twisted Edwards elliptic curve $\E$ with cofactor $2^c$, and a generator $\basePoint \in \E$ of a subgroup of points of prime order $q$. The order identifies the ring $\Int_q$ used for scalar values. Lastly they need to agree on the Purify parameters, in particular they choose a second elliptic curve $\E'$ over $\mathbb{F}_{q^2}$ and a base point $\basePoint'\in \E'$ which generates a group of points of order $q'$.
  
\subsection{Key Generation}\label{key-gen}
The parameters involved are:
\begin{center}
    \begin{tabular}{|ll|}
    	\hline
    	\textbf{Player 1} & \\ \hline
    	\textbf{Input:} & $\pk_3$  \\
     	\textbf{Private Output:} & $\omega_1, r'_1$\\
     	\textbf{Shared Secret:} 
     	& $\mathcal{D}$\\
    	\textbf{Public Output:} & $\rec_{1,3}$, $\rec_{2,3}$,\\
    	&$\ab$, $X_1,X_2$\\
    	\hline
    \end{tabular}
    \begin{tabular}{|ll|}
    	\hline
    	\textbf{Player 2} & \\ \hline
    	\textbf{Input:} & $\pk_3$ \\
     	\textbf{Private Output:} & $\omega_2, r'_2$\\
     	\textbf{Shared Secret:} 
     	& $\mathcal{D}$\\
    	\textbf{Public Output:} & $\rec_{1,3}$, $\rec_{2,3}$,\\
    	&$\ab$, $X_1,X_2$\\
    	\hline
    \end{tabular}
\end{center}
The protocol proceeds as follows:
\begin{enumerate}
	\item Secret key generation and communication:
	\begin{enumerate}[label=\alph*.]
	\item $P_i$ picks randomly $a_i, \sh_{3,i}, m_i \in \Int_q$, $r'_i\in\Int_{q'}$, and sets $\ab_i = a_i\basePoint$, ${\shb_{3,i} = \sh_{3,i}\basePoint}$, $\rtb_i = r'_i\basePoint'$, $\mb_i = m_i\basePoint$.
	\item $P_i$ computes $[\mathsf{KGC}_i, \mathsf{KGD}_i] = \Com((\ab_i, \shb_{3,i}, \rtb_i, \mb_i))$.
	\item $P_i$ sends $\mathsf{KGC}_i$ to $P_j$.
	\item $P_i$ sends $\mathsf{KGD}_i$ to $P_j$.
	\item $P_i$ gets $(\ab_j, \shb_{3,j}, \rtb_j, \mb_j) = \Ver(\mathsf{KGC}_j, \mathsf{KGD}_j)$, and saves the pairs $X_1=(\ab_1, \rtb_1)$ and $X_2=(\ab_2, \rtb_2)$.
	\end{enumerate}
	\item \label{FeldmanKeyGeneration} Feldman VSS and generation of $P_3$'s data:
	\begin{enumerate}[label=\alph*.]
	\item $P_i$ sets $f_i(x) = a_i + m_i x$ and computes $\sh_{i,j} = f_i(j)$ for $j\in\{1,2,3\}$.
	\item $P_i$ encrypts $\sh_{i,3}, \sh_{3,i}$ with $\pk_3$, let  $\rec_{i,3}$ be the pair of ciphertexts obtained.
	\item $P_i$ sends $\sh_{i,j}, \rec_{i,3}$ to $P_j$.
	\item If the asymmetric encryption algorithm supports DLOG verification, the encryption $\rec_{i, 3}$ is accompanied by two NIZKPs: the first one proves that the first ciphertext in $\rec_{i,3}$ is the encryption of the DLOG of $\shb_{i,3} = \ab_i + 3\mb_i$, the second NIZKP proves that the second ciphertext is the encryption of the DLOG of $\shb_{3,i}$.
	$P_i$ checks the NIZKPs attached to $\rec_{j, 3}$.
	\item $P_i$ checks, as in the Feldman-VSS described in \Cref{feldman}, the integrity and consistency of the shards $\sh_{j,i}$, verifying that ${\shb_{j,i} = \ab_j + i\mb_j}$, where $\shb_{j,i} = \sh_{j,i}\basePoint$.
	\item $P_i$ computes $x_i = \sh_{1,i} + \sh_{2,i} + \sh_{3,i}$ \label{ZKPDiscreteLogarithmRec}.
	\end{enumerate}
	\item $P_i$ proves in ZK the knowledge of $x_i$ using Schnorr's protocol of \Cref{Schnorr}.
	\item Public key and shards generation:
	\begin{enumerate}[label=\alph*.] 
	\item the public key is $\ab =\sum_{i=1}^3 \ab_i$, where $\ab_3=2\shb_{3,1} - \shb_{3,2}$, so that ${a_3 = 2\sh_{3,1} - \sh_{3,2}}$.
	From now on we will set $a = \sum_{i=1}^3 a_i$, obviously $a\basePoint = \ab$.
	\item $P_1$ computes $\omega_1 = 2x_1$, while $P_2$ computes ${\omega_2 = -x_2}$.
	\item $P_i$ computes the common secret $\mathcal{D}=\sh_{3,i}\shb_{3,j}$.
	\end{enumerate}
\end{enumerate}
\begin{oss}\label{ossu3}
  We define $a_3 = 2\sh_{3,1} - \sh_{3,2}$ because we need to be consistent with the Feldman-VSS protocol.
  Indeed, suppose that $\sh_{3,2}$ and $\sh_{3,1}$ are valid shards of a Feldman-VSS protocol where the secret is $a_3$.
  In this way we have that ${\sh_{3,2} = a_3+2m_3}$ and $\sh_{3,1} = a_3+m_3$, so: $$2\sh_{3,1}-\sh_{3,2} = 2a_3+2m_3-a_3-2m_3=a_3.$$
  Note that $\ab_3 = a_3\basePoint$ can be computed by both $P_1$ and $P_2$, but $a_3$ cannot.
 \end{oss}
 
\subsection{Signature Algorithm}
This protocol is used by two players, called $P_A$ and $P_B$, to sign messages.
$P_1, P_2$, and $P_3$ take the role of either $P_A$ or $P_B$ depending on the situation, see \Cref{ordinarysignature,recoverysignature}.

The participants agree on a message $M$ to sign and the goal of this protocol is to produce a valid EdDSA signature $(\mathcal{R},S)$ for the public key $\ab$.\\
The parameters involved are:
\begin{center}
    \begin{tabular}{|ll|}
    	\hline
    	\textbf{Player $A$} & \\ \hline
    	\textbf{Input:} & $M, \omega_A, \ab, r'_A$ \\
    	& $X_A,X_B$\\
    	\textbf{Public Output:} &  $(\mathcal{R},S)$\\
    	\hline
    \end{tabular}
    \begin{tabular}{|ll|}
    	\hline
    	\textbf{Player $B$} & \\ \hline
    	\textbf{Input:} & $M, \omega_B, \ab, r'_B$ \\
    	& $X_A,X_B$\\
    	\textbf{Public Output:} &  $(\mathcal{R},S)$\\
    	\hline
    \end{tabular}
\end{center}
The protocol works as follows:
\begin{enumerate}
	\item Generation of $\mathcal{R}$:\label{SignaturePrimoPunto}
	\begin{enumerate}[label=\alph*.]
	\item $P_i$ computes $\mathtt{K}=H(X_A,X_B)$. \label{SignatureDoveSiCalcolaK}
	\item $P_i$ computes $\V'=H_{\mathtt{Pur}}(\mathtt{K},M)\in\E'$.
	\item $P_i$ computes $r_i=f(r'_i\V')$.
	\item $P_i$ computes $\rb_i=r_i\basePoint$.
	\item $P_i$ sends $\rb_i$ to $P_j$ alongside a non-interactive zero-knowledge proof that it is correct given $\rtb_i$ (see~\Cref{ZKPpurify}).
	\item $P_i$ checks the correctness of the value $\rb_j$ received by verifying the attached NIZKP.
	\item $P_i$ computes $\rb = \rb_A + \rb_B$.
	\end{enumerate}
	\item Generation of $S$:
	\begin{enumerate}[label = \alph*.]
	\item $P_i$ computes $S_i = r_i + \omega_i H(\mathcal{R}\parallel \ab\parallel M)$
	\item $P_i$ sends $S_i$ to $P_j$.
	\item $P_i$ computes $S = S_A+S_B$.
	\end{enumerate}
	\item $P_i$ checks that $S\basePoint = \mathcal{R}+H(\mathcal{R}\parallel \ab\parallel M)\ab$.
\end{enumerate}
If any check fails the protocol aborts, otherwise the output signature is $(\mathcal{R},S)$.

\subsection{Ordinary Signature}\label{ordinarysignature}
This is the case where $P_1$ and $P_2$ wants to sign a message $m$. They run the signature algorithm with the following parameters (suppose wlog that $P_1$ plays the roles of $P_A$ and $P_2$ of $P_B$):
\begin{center}
    \begin{tabular}{|ll|}
    	\hline
    	\textbf{Player $A$} & \\ \hline
    	\textbf{Input:} & $\omega_1, r'_1,$\\ &$M, \ab, X_1, X_2$\\
    	\textbf{Public Output:} &  $(\mathcal{R},S)$\\
    	\hline
    \end{tabular}
    \begin{tabular}{|ll|}
    	\hline
    	\textbf{Player $B$} & \\ \hline
    	\textbf{Input:} & $\omega_2, r'_2,$\\ &$M, \ab, X_1, X_2$\\
    	\textbf{Public Output:} &  $(\mathcal{R},S)$\\
    	\hline
    \end{tabular}
\end{center}
\subsection{Recovery Signature}\label{recoverysignature}
If one between $P_1$ and $P_2$ is unable to sign, then $P_3$ has to come back online and a recovery signature is performed.

We have to consider two different cases, depending on who is offline.
First we consider the case in which $P_2$ is offline, therefore $P_1$ and $P_3$ sign.\\
The parameters involved are:
\begin{center}
    \begin{tabular}{|ll|}
    	\hline
    	\textbf{Player 1} & \\ \hline
    	\textbf{Input:} & $\omega_1, r'_1$,  \\
    	& $M, \ab, X_1, \rec_{1,3}, \rec_{2,3}$\\
    	\textbf{Public Output:} &  $(\mathcal{R},S)$\\
    	\hline
    \end{tabular}
    \begin{tabular}{|ll|}
    	\hline
    	\textbf{Player 3} & \\ \hline
    	\textbf{Input:} & $\sk_3,$ \\
    	& $M$\\
    	\textbf{Public Output:} &  $(\mathcal{R},S)$\\
    	\hline
    \end{tabular}
\end{center}
The workflow in this case is:
\begin{enumerate}
	\item Communication:
	\begin{enumerate}[label=\alph*.]
	\item $P_1$ contacts $P_3$ and sends $\ab, \rec_{1,3}, \rec_{2,3},X_1$.
	\item $P_3$ decrypts everything with the private key $\sk_3$ to recover the values $\sh_{1,3}$, $\sh_{3,1}$, $\sh_{2,3}$, $\sh_{3,2}$.
	\item $P_3$ computes $a_3= 2\sh_{3,1} -  \sh_{3,2}$ and $\ab_3 = a_3\basePoint$.
	\item\label{soloUnaVoltaInizio13} $P_3$ picks randomly $r'_3 \in \Int_{q'}$ and computes $\rtb_3=r'_3\basePoint'$.
	\item\label{soloUnaVoltaFine13} $P_3$ sends $X_3 = (\ab_3, \rtb_3)$ to $P_1$.
	\end{enumerate}
	\item $P_3$'s key creation:
	\begin{enumerate}[label=\alph*.]
	\item $P_3$ computes $x_3 = \sh_{1,3} + \sh_{2,3} + 2\sh_{3,2} - \sh_{3,1}$.
	\item $P_i$ proves in ZK the knowledge of $x_i$ using Schnorr's protocol (note that $x_1 = \omega_1/2$).
	\end{enumerate}
	\item Signature generation:
	\begin{enumerate}[label = \alph*.]
	\item $P_1$ computes $\tilde{\omega}_1 = \frac{3}{4} \omega_1$.
	\item $P_3$ computes $\omega_3 = -\frac{1}{2} x_3$.
	\item $P_1$ and $P_3$ perform the Signature Algorithm as $P_A$ and $P_B$ respectively, where $P_1$ uses $\tilde{\omega}_1$ instead of $\omega_A$ and $X_1$ instead of $X_A$, while $P_3$ uses $\omega_3$ in place of $\omega_B$ and $X_3$ in place of $X_B$ (the other parameters are straightforward).
	\end{enumerate}
\end{enumerate}

We consider now the second case in which $P_1$ is offline, therefore $P_2$ and $P_3$ sign.
The parameters involved are:

\begin{center}
    \begin{tabular}{|ll|}
        \hline
    	\textbf{Player 2} & \\ \hline
    	\textbf{Input:} & $\omega_2, r'_2$,  \\
    	& $M, \ab, X_2, \rec_{1,3}, \rec_{2,3}$\\
    	\textbf{Public Output:} &  $(\mathcal{R},S)$\\
    	\hline
    \end{tabular}
    \begin{tabular}{|ll|}
    	\hline
    	\textbf{Player 3} & \\ \hline
    	\textbf{Input:} & $\sk_3,$ \\
    	& $M$\\
    	\textbf{Public Output:} &  $(\mathcal{R},S)$\\
    	\hline
    \end{tabular}
\end{center}
The first two steps are identical to the previous case (for the ZKP of $x_2$ note that $x_2 = -\omega_2$).
\begin{enumerate}
\setcounter{enumi}{2}
    \item The Signature generation step proceeds as follows:
\begin{enumerate}[label=\alph*.]
	\item $P_2$ computes $\tilde{\omega}_2 = -3 \omega_2$.
	\item $P_3$ computes $\omega_3 = -2 x_3$.
	\item $P_2$ and $P_3$ perform the Signature Algorith as $P_A$ and $P_B$ respectively, where $P_2$ uses $\tilde{\omega}_2$ instead of $\omega_A$ and $X_2$ instead of $X_A$ and $P_3$ uses $\omega_3$ in place of $\omega_B$ and $X_3$ in place of $X_B$ (the other parameters are straightforward).
\end{enumerate}
\end{enumerate}

\begin{oss}
    $\rtb_3$ could be generated and published ahead of time (e.g. during the setup phase), and used for all subsequent recovery signatures.
    In this case, after $P_3$ computes $X_3$ for the first time then the value is fixed for all the successive executions.
    
    The reasons why it is necessary to have a different $X_i$ for each player will be more clear later, during the security discussion of the protocol in Section \ref{security}.
\end{oss}

 \subsection{Key Derivation}\label{key-derivation}
 In order to perform the key derivation we need a derivation index $i$ and the common secret $\mathcal{D}$ created during the Key Generation protocol.
\\The derivation is performed as follows:
\begin{itemize}
    \item $P_1$ and $P_2$ perform the key derivation:
    \begin{itemize}
        \item[$\diamond$] $\omega_1 \to \omega_1^i = \omega_1 + 2H(\mathcal{D}\parallel i),$
        \item[$\diamond$] $\omega_2 \to \omega_2^i = \omega_2 - H(\mathcal{D}\parallel i);$
    \end{itemize}
    \item $P_1$ and $P_3$ perform the key derivation:
    \begin{itemize}
        \item[$\diamond$] $\omega_1 \to \omega_1^i = \omega_1 + \frac{3}{2}H(\mathcal{D}\parallel i),$
        \item[$\diamond$] $\omega_3 \to \omega_3^i = \omega_3 - \frac{1}{2}H(\mathcal{D}\parallel i);$
    \end{itemize}
    \item $P_2$ and $P_3$ perform the key derivation:
    \begin{itemize}
        \item[$\diamond$] $\omega_2 \to \omega_2^i = \omega_2 + 3H(\mathcal{D}\parallel i),$
        \item[$\diamond$] $\omega_3 \to \omega_3^i = \omega_3 -2H(\mathcal{D}\parallel i);$
    \end{itemize}
    \item the public key is always updated like this: ${\ab \to \ab^i = \ab+H(\mathcal{D}\parallel i)\basePoint}.$
\end{itemize}
\begin{oss}
    We observe that the algorithm outputs valid keys, such that, for example:
    $$ (\omega_1^i+\omega_2^i)\basePoint = \ab^i.$$
    Since $(\omega_1^i+\omega_2^i) = \omega_1 + \omega_2 + H(\mathcal{D}\parallel i)$ we have that:
    $$ (\omega_1^i+\omega_2^i)\basePoint = (\omega_1 + \omega_2 + H(\mathcal{D}\parallel i))\basePoint = \ab + H(\mathcal{D}\parallel i)\basePoint = \ab^i.$$
    With the same procedure we can prove that also the other pairs of derived keys are consistent.
\end{oss}

\section{Security Proof} \label{security}
As customary for digital signature protocols, we state the security of our scheme as an \emph{unforgeability} property, defined as follows:
\begin{defn}
	We say that a $(t, n)$-threshold signature scheme is unforgeable if no malicious adversary who corrupts at most $t-1$ players can produce with non-negligible probability the signature on a new message $m$, given the view of \textbf{Threshold-Sign} on input messages $m_1,...,m_k$ (which the adversary adaptively chooses), as well as the signatures on those messages.
\end{defn}

Referring to this definition, the security of our protocol derives from the following theorem, whose proof is the topic of this section:
\begin{thm} \label{teoEdDSA}
    Let $\xi : \{0,1\}^{2b} \rightarrow \Int_q$ be the encoding that maps bit-strings into elements of $\Int_q$ via \emph{little-endian} encoding and reduction modulo $q$, let ${\pi : \{0,1\}^{2b} \rightarrow \{0,1\}^{n}}$ be the function that truncates a bit-string to $n$ bits: ${\pi(h) = h_0\parallel \ldots\parallel h_{n-1}}$.
    Then, assuming that:
	\begin{itemize}
	    \item $H'$ is a cryptographic hash function such that $H=\xi \circ H'$ is a good PRNG as per~\Cref{good-PRNG}, and $\pi \circ H'$ is a strong PRNG as per~\Cref{strong-PRNG};
		\item the EdDSA signature scheme with parameters $(p, b, H', c, n, a, d, \basePoint, q)$ is unforgeable;
		\item $\Com, \Ver$ is a non-malleable commitment scheme as defined in~\Cref{commitments};
 	    \item the Decisional Diffie Hellman Assumption defined in~\Cref{DDH} holds for both the curves $\E$ and $\E'$;
		\item the encryption algorithm used by $P_3$ is IND-CPA, as per \Cref{IND-CPA};
	\end{itemize}
	our threshold protocol built with the hash function $H$ is unforgeable.
\end{thm}

The proof will use a classical game-based argument, our goal is to show that if there is an adversary $\ab$ that forges the threshold scheme with a non-negligible probability $\e>\lambda(\mathtt{k})^{-t}$, where $\mathtt{k}$ is the security parameter, for a polynomial $\lambda(x)$ and $t>0$, then we can build a forger $\F$ that forges the centralized EdDSA scheme with non-negligible probability as well.

Since the algorithm presented is a $(2,3)$-threshold signature scheme, the adversary will control one player and $\F$ will simulate the remaining two.
Since the role of $P_3$ is different than those of $P_2$ and $P_3$, we have to consider two distinct cases: one where $\adv$ controls $P_3$ and one where $\adv$ controls one between $P_1$ and $P_2$ (whose roles are symmetrical).
The second case is way more interesting and difficult, so it will be discussed first, and for now we suppose without loss of generality that $\adv$ controls $P_2$.

The adversary $\adv$ interacts in our protocol as follows: it first participates in the key generation protocol to generate a public key $\ab$ for the threshold scheme, then it requests the signature on some messages $m_1,...,m_l$.
During this phase it can participate in the signature generation or it can query for signatures generated by $P_1, P_3$.
Eventually the adversary outputs a message $m\ne m_i\; \forall i$ and a valid signature on $m$ with probability at least $\e$.
If we denote with $\tau_\adv$ the adversary's tape and with $\tau_i$ the tape of the honest player $P_i$ we can write:
  \begin{equation}\label{avversario}
  	 \Pb_{\tau_i, \tau_\adv}[\adv(\tau_\adv)_{ P_i(\tau_i)} = \mathtt{forgery}] \ge \e,
  \end{equation}
  where $\Pb_{\tau_i, \tau_\adv}$ means that the probability is taken over the random tape $\tau_\adv$ of the adversary and the random tape $\tau_i$ of the honest player, while $\adv(\tau_\adv)_{P_i(\tau_i)}$ is the output of the iteration between the adversary $\adv$, running on tape $\tau_\adv$, and the player $P_i$, running on tape $\tau_i$ .
  
\begin{defn}[Good Tape]
  We say that an adversary's random tape $\tau_\adv$ is good if:
  \begin{equation}\label{defbuono}
    \Pb_{\tau_i}[\adv(\tau_\adv)_{ P_i(\tau_i)} = \mathtt{forgery}] \ge \frac{\e}{2}.
  \end{equation}
\end{defn}
\noindent
Now we have the following Lemma, introduced in~\cite{gennaro}:
\begin{lem} \label{goodTape}
	If $\tau_\adv$ is a tape chosen uniformly at random, the probability that it is a good one is at least $\frac{\e}{2}$.
\end{lem}
\begin{proof}
	In the proof we will simplify the notation writing $\adv(\tau_\adv,\tau_i) = \mathtt{forgery} $ instead of ${\adv(\tau_\adv)_{P_i(\tau_i)} = \mathtt{forgery}}$. In the context of this proof, we will write $b$ to identify a good tape, while $c$ will be a bad one.
	We can rewrite \Cref{avversario} in this way:
	\begin{align}
	    A &= \Pb_{\tau_i,\tau_\adv}(\tau_\adv = b, A(\tau_\adv, \tau_i)= \mathtt{forgery}) +  \Pb_{\tau_i,\tau_\adv}(\tau_\adv = c, A(\tau_\adv, \tau_i) = \mathtt{forgery})\nonumber\\
	    &= \Pb_{\tau_i,\tau_\adv} (\tau_\adv = b)  \Pb_{\tau_i,\tau_\adv}(A(\tau_\adv, \tau_i) = \mathtt{forgery}\vert\tau_\adv = b)\nonumber\\
	    &\phantom{=}\, +  \Pb_{\tau_i,\tau_\adv}(\tau_\adv = c)  \Pb_{\tau_i,\tau_\adv}(A(\tau_\adv, \tau_i) = \mathtt{forgery} \vert\tau_\adv = c).
	\end{align}
	Trivially we have that $\Pb_{\tau_i,\tau_\adv}(A(\tau_\adv, \tau_i) = \mathtt{forgery}\vert\tau_\adv = b) < 1,$ and from the definition of good tape in equation \ref{defbuono} we get:
	\begin{equation}
	    \Pb_{\tau_i,\tau_\adv}(A(\tau_\adv, \tau_i) = \mathtt{forgery}\vert\tau_\adv = c) < \frac{\e}{2}.
	\end{equation}
	Now we want to solve for $x =  \Pb_{\tau_i,\tau_\adv}(\tau_\adv = b)$, so we get:
	\begin{equation}
	    \e \le A < x \cdot 1 + (1 - x) \cdot \frac{\e}{2} = x \left(1 - \frac{\e}{2}\right) + \frac{\e}{2},
	\end{equation}
that leads us to the conclusion:
\begin{equation}
    x \ge \frac{\e - \frac{\e}{2}}{1 - \frac{\e}{2}} \ge \frac{\e}{2-\e} \ge \frac{\e}{2}.
\end{equation}
\end{proof}
From now on we will suppose that the adversary is running on a good random tape.

\
\\First of all we need to deal with the key generation algorithm.
The simulator $\F$ plays the role of $P_1$ and $\adv$ plays the role of $P_2$.
Before starting the simulation, $\F$ receives from its challenger a public key for the IND-CPA encryption algorithm and an EdDSA public key.
The goal is to trick $\adv$ in order to force the public key output by the multi-party computation to match the EdDSA key given by the challenger.
\\
\\
The simulation works as follows:
\begin{enumerate}
	\item $\F$ receives from the challenger an EdDSA public key $\ab_c$ and the public encryption key $\pk_3$.
	
	\item\label{Edrepeatfirst} $P_i$ picks randomly  $a_i, \sh_{3,i}, m_i \in \Int_q$, $r'_i\in\Int_{q'}$, and sets $\ab_i = a_i\basePoint$, ${\shb_{3,i} = \sh_{3,i}\basePoint}$, $\rtb_i = r'_i\basePoint'$, $\mb_i = m_i\basePoint$.
	
	\item $P_i$ computes $[\mathsf{KGC}_i, \mathsf{KGD}_i] = \Com((\ab_i, \shb_{3,i}, \rtb_i, \mb_i))$.
	
	\item $P_2$ sends $\mathsf{KGC}_2$ to $P_1$.
	
	\item\label{Edcommit} $P_1$ sends $\mathsf{KGC}_1$ to $P_2$.
	It is important that $P_2$ sends its commitment before $P_1$, see~\Cref{rush}.
	
	\item $P_i$ sends $\mathsf{KGD}_i$ to $P_j$.
	
	\item $P_i$ gets $(\ab_j, \shb_{3,j}, \rtb_j, \mb_j) = \Ver(\mathsf{KGC}_j, \mathsf{KGD}_j)$.
	
	\item At this point $\F$ knows all the parameters involved in the computation of $\ab$, the first part of the key.
	So it rewinds $\adv$ to the step \ref{Edcommit}, after the commitment of $\adv$, with the aim to make $\ab=\ab_c$.
	
	\item $\F$ computes $\hat{\ab} = \ab_c - \ab_2 - 2 \shb_{3,1} + \shb_{3,2}$.
	
	\item $\F$ picks randomly $\sh_{1,2} \in \Int_q$ and computes $\widehat{\mb} = \frac{1}{2}(\sh_{1,2}\basePoint - \hat{\ab})$ to simulate the VSS (since $\F$ is not able to compute the random polynomial $f(x)$) as explained in \Cref{feldman}.
	
	\item $\F$ computes the commitment $[\hat{\mathsf{KGC}}_i,\hat{\mathsf{KGD}}_i] = \Com((\hat{\ab}, \shb_{3,i}, \rtb_i, \widehat{\mb}))$ and sends it to $\adv$ as $P_1$.
	
	\item $P_1$ sends $\hat{\mathsf{KGD}}$ to $P_2$.
	
	\item $P_i$ picks randomly $\sh_{1,3} \in \Int_q$ and encrypts $\sh_{1,3}, \sh_{3,1}$ with $\pk_3$, obtaining $\rec_{i,3}$ ($\F$ has to simulate the NIZKPs if the encryption supports DLOG verification).
	
	\item $P_i$ sends $\sh_{i,j}, \rec_{i,3}$ to $P_j$.

	\item\label{Edrepeatlast}
	Since $\F$ does not know the discrete logarithm of $\hat{\ab}$ it can not compute $x_1$, so it simulates the ZKP with $\adv$.
	
	\item $P_2$ can calculate $x_2$ and execute the ZKP, from which $\F$ extracts the value of $x_2$.
	
	\item $P_i$ can compute the key $\ab$, moreover $P_2$ can compute $\omega_2$ (for $\F$ it is impossible, since it does not know $x_1$).
\end{enumerate}

\begin{oss} \label{rush}
	In the simulation it is crucial that the adversary broadcasts $\mathsf{KGC}_2$ before $\F$.
	Inverting the order will cause this simulation to fail, since after the rewind $\adv$ could change its commitment.
	Due to the non-malleability property we are assured that $\adv$ can not deduce anything about the content of these commitments, but nevertheless it could use it as a seed for the random generation of its values.
	In this case $\F$ guesses the right $\hat{\ab}$ only with probability $\frac{1}{q}$ where $q$ is the size of the group, so the expected time is exponential.

	It is possible to swap the order in the first step using an equivocable commitment scheme with a secret trapdoor.
	In this case we only need to rewind at the decommitment step, we change $\mathsf{KCD}_1$ in order to match $\hat{\ab}$ and $\hat{\mb}$.
	In this way we could prove the security of the protocol also in the presence of a \emph{rushing adversary}, but we need an additional hypothesis regarding the commitment scheme.
\end{oss}

\begin{lem}  \label{lemma1Eddsa}
	The simulation terminates in expected polynomial time and it is indistinguishable from the real protocol.
\end{lem}
\begin{proof}
	Since $\adv$ is running on a good random tape we know that it will correctly decommit with probability at least $\frac{\e}{2} \ge \lambda(k)^{-t}$, then we need to rewind only a polynomial number of times.
	$\F$ does not know the discrete logarithm of $\hat{\ab}$ and so it needs to perform a ``fake'' Feldman-VSS.
	This is indistinguishable from a real Feldman-VSS since they have the same distribution, as shown in \Cref{feldman}.
	
	The Schnorr protocol can be perfectly simulated as shown in Appendix \ref{Schnorr} due to the Decisional Diffie-Hellman Assumption.
	
	The security properties of the encryption algorithm assure that the simulated $\rec_{1,3}$ is indistinguishable from the real one (thanks to IND-CPA for the ciphertext and, if necessary, the zero-knowledge property for the NIZKP).
	
	There are no other differences between the real protocol and the simulated one, so the simulation is indeed indistinguishable.
\end{proof}
\begin{lem} \label{lemma2eddsa}
	For a polynomially large fraction of inputs $\ab_c$ the simulation terminates with output $\ab_c$, except with negligible probability.
\end{lem}
\begin{proof}
	First we prove that if the simulation terminates correctly (i.e. with output different from $\perp$) then it terminates with output $\ab_c$ except with negligible probability.
	
	This is a consequence of the non-malleability property of the commitment scheme.
	Indeed, if $\adv$ correctly decommits twice it must do so with the same string, no matter what $P_1$ decommits to (except with negligible probability).
	Therefore, due to our choice for $\hat{\ab}$ we have that the output is $\ab_c$.
	
	Now we prove that the simulation ends correctly for a polynomially large fractions of input.
	Since $\adv$ is running on a good random tape, it decommits correctly for at least $\frac{\e}{2}>\lambda{k}^{-t}$ inputs.
	Moreover, since $\ab_c$ is chosen uniformly at random and $\hat{\ab} = \ab_c - \ab_2 - 2 \shb_{3,1} + \shb_{3,2}$ is fully determined after the rewind, we have that $\hat{\ab}$ has also uniform distribution, then we can conclude that for at least a fraction $\frac{\e}{2}>\lambda{k}^{-t}$ of input the protocol will correctly terminate.
\end{proof}

Now we have to deal with the ordinary signature algorithm.
Here $\F$ can fully predict what $\adv$ will output and then it can choose its shards in order to match the signature it received from its oracle.
Comparing to the proofs of ECDSA threshold protocols in~\cite{gennaro2018fast,battagliola2020threshold} we do not need to make a distinction between semi-correct and non-semi-correct executions, since we can always provide a perfect simulation that ends with the desired result (except with negligible probability).

It is important to  remember that $\F$ does not know the secret key of $P_1$ but it knows everything about $P_2$, since it was able to extract the secret values during the ZKPs.
\\
\\
The simulation works as follows:
\begin{enumerate}
	\item $\adv$ chooses a message $M$ to sign.
	
	\item $\F$ queries its signing oracle for a signature for $M$ corresponding to the public key $\ab$, and gets $(\mathcal{R}_f,S_f)$.
	\label{RewindSignature}
	
	\item $P_i$ follows the protocol normally and computes $\rb = \rb_1+\rb_2$. We can notice that $\F$ is able to follow the protocol normally since it knows $r'_1$ and therefore can compute $\rb_1$.
	
	\item $\F$ computes $\hat{\rb} = \rb_f - \rb_2$ and rewinds the adversary at the end of Step~\ref{RewindSignature}.
	
	\item $\F$ sets $\rb_1 = \hat{\rb}$.
	
	\item $P_i$ follows the protocol normally to get $\rb = \rb_1+\rb_2$.
	
	\item $\F$ simulates the ZKP using as input $\rtb_1$ and $\hat{\rb}$.
	
	\item From the ZKP given by the adversary on behalf of $P_2$, $\F$ is able to extract $r'_2$ from the adversary, and therefore also $r_2$.
	
	\item Since $\F$ knows both $r_2$ and $\omega_2$, $P_1$ can compute $S_1$ as: $\linebreak{S_1 = S_f - r_2 - \omega_2 H(\rb\parallel \ab\parallel M)}$.
	
	\item $P_i$ follows the protocol normally to get $S=S_1+S_2$.
	
	\item $P_i$ checks that $ S\basePoint =  \mathcal{R}+H(\mathcal{R}\parallel \ab\parallel M)A$.
\end{enumerate}
If any check fails the protocol aborts, otherwise the output signature is $(\rb,S)$.

\begin{lem}
	If Purify is secure in the sense of \Cref{LemmaPurify}, then the protocol above is a perfect simulation of the centralized one and terminates correctly with output $(\mathcal{R}_f,S_f)$.
\end{lem}
\begin{proof}
	The differences between the simulation and the real protocol is that $\F$ does not know the secret key $\omega_1$ when computing $S_1$, and the computation of $\rb_1$ uses a different PRF.
	
	The lack of knowledge of the secret key is not a problem since $\F$ is able to retrieve the correct values to output knowing ahead of time what $\adv$ should output.
	
	About the different method used to compute $\rb_1$, notice that the assumptions on Purify mean that in a real execution of the protocol $r_i$ has a distribution that is indistinguishable from the uniform distribution over $\Int_q$, i.e. the distribution of $\rb_i$ .
	In the simulation $\rb_1 = \rb_f - \rb_2$, where the distribution of $\rb_2$ is indistinguishable from the uniform distribution over the group generated by $\basePoint$ (as in the real protocol) and $\rb_f$ comes from the centralized EdDSA oracle.
	From our assumptions on the hash function used in EdDSA, the distribution of $\rb_f$ is also indistinguishable from the uniform distribution over the group generated by $\basePoint$, consequently so is the distribution of $\rb_1$.
	
	It is straightforward that if the protocol terminates it will do so with output $(\rb,S)=(\rb_f,S_f)$, in fact if $\adv$ does not act honestly the check in the last step will fail with high probability.
	
	Finally the only other way that the protocol does not terminate is when the NIZKP of $\rb_2$ does not verify.
	In this case the simulation simply aborts, like in a real execution of the protocol.
\end{proof}

Now we have to deal with the recovery signature.
Since the core algorithm remains the same we can use the proof just explained, we only need to change the setup phase during which the third player recovers its secret material.

First of all we will examine what happens if $\adv$ controls one between $P_1$ or $P_2$ and $\F$ controls $P_3$. Then we will deal with the case in which $\adv$ controls $P_3$, that will be easier since the whole enrollment phase can be avoided.

Trivially if $\adv$ asks for a recovery signature between the two honest parties $\F$ can simply ask its oracle and output whatever it received from the oracle. So we can limit ourselves to deal with the case where $\adv$ participates in the signing process.

\noindent
If $\adv$ controls $P_2$ the simulation proceeds as follow:
\begin{enumerate}
	\item $P_2$ sends to $P_3$ $\ab, X_2, \rec_{1,3}, \rec_{2,3}$.
	
	\item $\F$ has participated in the Key-Generation  phase, so knows $\shb_{3,1}$ and  $\shb_{3,2}$, so can compute on behalf of $P_3$ the value $\ab_3 = 2\shb_{3,1} -  \shb_{3,2}$.
	
	\item $P_3$ picks randomly $r'_3 \in \Int_{q'}$ and computes $\rtb_3=r'_3\basePoint'$.
	
	\item $P_3$ sends $X_3 = (\ab_3, \rtb_3)$ to $P_1$.
	
	\item Note that $P_3$ can not decrypt the values received in the first step, so it simulates the ZKP about $x_3$, conversely $\F$ can extract $x_2$ from $P_2$.
	
	\item $P_2$ computes $\tilde{\omega}_2 = -3 \omega_2$.
	$P_3$ can not compute its secret key $\omega_3$, but this is not a problem as we explained before.
	
	\item They perform the signing algorithm with the above simulation.
    Also in this case $\F$ does not know its own secret key, but we remark that this is fine since it knows $P_2$'s secrets and it can use the signing oracle.
\end{enumerate}
In the same way we can deal with the case of $\adv$ controlling $P_1$.

\vspace{.5cm}

\noindent
Now we have to deal with the last case, i.e. when $P_3$ is the dishonest party.

\vspace{.5cm}

During the enrollment phase $\F$ can produce random shards, which will be sent to $P_3$ during the recovery signature phase, and output the public key given by the EdDSA challenger.
These random shards simulate correctly the protocol for the properties of the secret sharing.
In fact the only difference is that once again $\F$ does not know the corresponding secret keys of one between $P_1$ and $P_2$ (one player's keys can be chosen freely, but the others are forced by the challenge public keys), but as before this is not a problem because, thanks to the oracle and the secrets it extracts from $P_3$, $\F$ can simulate signatures with the same simulation described above.

\vspace{.5cm}

Now we are ready to prove~\Cref{teoEdDSA}.
\begin{proof}
	As we previously proved, our simulator produces a view of the protocol indistinguishable from the real one for the adversary, so $\adv$ will produce a forgery with the same probability as in a real execution.
	Then the probability of success of our forger $\F$ is at least $\frac{\e^3}{8}$, since $\F$ has to succeed in
	\begin{itemize}
		\item choosing a good random tape for $\adv$, whose probability is at least $\frac{\e}{2}$, as shown in~\Cref{goodTape},
		\item hitting a good public key, whose probability also is at least $\frac{\e}{2}$ as shown in~\Cref{lemma1Eddsa} and~\Cref{lemma2eddsa}.
	\end{itemize}
	Under those conditions $\adv$ successfully produces a forgery with probability at least $\frac{\e}{2}$ as per \Cref{defbuono}.
	Under the security of the EdDSA signature scheme, the probability of success of $\F$ must be negligible, which implies that $\e$ must negligible too, contradicting the hypothesis that $\adv$ has non-negligible probability of forging the scheme.
\end{proof}

\begin{oss}
    As we said in the Signature Algorithm description, at point \ref{SignaturePrimoPunto}\ref{SignatureDoveSiCalcolaK}, we need to have a different $\mathtt{K}$ for each pair of signers, otherwise an adversary having access to all the messages exchanged by the honest parties could steal the secret key.
    It suffices to ask for the signature of the same message, first the signature is performed by the honest parties, then by the adversary and a honest party, as explained in Section 4 of \cite{Purify}.
\end{oss}

\subsection{Resilience of the recovery}\label{SottosezioneResilienzaRecovery}
In our security analysis we focused on the unforgeability of the signature, however with an offline party another security aspect is worthy of consideration: the resiliency of recovery in the presence of a malicious adversary. Of course if the offline party is malicious and unwilling to cooperate there is nothing we can do about it, however the security can be strengthened if we consider that one of the online parties may corrupt the recovery material. In this case a generic CPA asymmetric encryption scheme is not sufficient to prevent malicious behaviour, because we need a verifiable encryption scheme that allows the parties to prove that the recovery material is consistent, just like they prove that they computed the shards correctly.

In particular we need an encryption scheme that support DLOG verification as explained in point \ref{FeldmanKeyGeneration}\ref{ZKPDiscreteLogarithmRec} of the Key-Generation algorithm.
A suitable candidate is a variant of the  Cramer–Shoup cryptosystem presented in \cite{VerifiableEncryption}.
This algorithm equipped with a ZKP that allow the sender to prove that the plaintext he encrypted is the discrete logarithm of a public value.
In particular, since the protocol is a three step ZKP with special soundness, completness and honest-verifier zero knowledge it is possible to build a non-interactive ZKP using the Fiat-Shamir heuristic.

\section{Conclusions}\label{conclusions} 
Although decentralized signature algorithms have been known for a while, we are aware of only few proposals for algorithms that are able to produce signatures indistinguishable from a standard one.
The protocol described in this work is, as far as we know, the first example of threshold multi-signature allowing the presence of an off-line participant during key-generation and whose signatures are indistinguishable from EdDSA ones. 

The approach we have taken is very similar to the one presented in \cite{ECDSA2020threshold} and ~\cite{gennaro2018fast}, although there are some key differences between the works. 
First of all our main idea is to have two active participants to simulate the action of the third one.
This step is possible due to the uniqueness property of polynomial interpolation that gives a bijection between points and coefficients, which allows us to ``invert'' the generation of the shares, thanks to the preserved uniform distribution  in $\Int_p$. These shares are later recovered by the offline party exploiting an asymmetric encryption scheme.
A second difference is that we have managed to avoid equivocable commitments, under the assumption that in some specific steps (see~\Cref{rush}) we can consider the adversary not to be rushing.

The focus of this work was to shift away from DSA-like protocols and study a more recent standard like EdDSA.
We remark that ECDSA is more suited to be used in a multi-party environment: the absence of hash functions to be computed on private data allows a more straightforward adaption to a multi-party setting.
Indeed, a joint computation of a standard hash function is difficult in a reasonable time.
Therefore, when creating an EdDSA-compatible threshold multi-signature scheme there is the necessity of working around this issue. Our solution is to build a variant of the EdDSA protocol, whose outputs are indistinguishable from those of the original scheme (and therefore it preserves the security properties), thus avoiding joint hash computations.

On the other hand, multi-party EdDSA requires less message exchanges between the participants than the amount required for the threshold ECDSA protocol in \cite{ECDSA2020threshold}, since less checks are needed to avoid malicious computations.

A last remark worth to be mentioned is in the work-around made on the Zero-Knowledge proofs of our EdDSA scheme (as explained in \cite{ristretto}), which are required to work around the usage of elliptic curves whose group of points does not have prime order.

Similarly to its ECDSA counterpart, in order to guarantee the security of the signature itself against black-box adversaries, the protocol involves a large utilization of ZKPs.
Despite of the consequent drawbacks in terms of efficiency, our protocols has been succesfully implemented and adopted in the management of Libra wallets~\cite{libra_2020}.

Other future research steps involve the generalization to $(t,n)$-threshold schemes with more than one offline party and the extension of our notion of security.
Although our protocol is susceptible to DOS attacks on the offline party, there are many ways to overcome this apparent weakness, such as the distribution of the role of the Recovery party to multiple servers or the generalization of our scheme to more than three parties.

\section*{Acknowledgments}
The core of this work is contained in the first author's MSC thesis supervised by the second and fourth author.
\\
The authors would like to thank Conio s.r.l. and its co-CEO Vincenzo di Nicola for their support. We also thank Gaetano Russo, Federico Mazzone, and Zsolt Levente Kucsv\'an that worked on the implementation and provided valuable feedback.


\section*{Declarations}
The authors have no relevant financial or non-financial interests to disclose.
The authors have no conflicts of interest to declare that are relevant to the content of this article.
The second and the third authors are members of the INdAM Research group GNSAGA.
The first author acknowledges support from TIM S.p.A. through the PhD scholarship.
The authors have no financial or proprietary interests in any material discussed in this article.

\bibliography{sn-bibliography}

\begin{appendices}
\section{Zero Knowledge Proofs}

\subsection{Schnorr Protocol}\label{Schnorr}
The Schnorr Protocol is a zero-knowledge proof for the discrete logarithm.
\\Let $\G$  be a group of prime order $q$ with generator $\basePoint$.
Let $\kb\in \G$ be a random element in $\G$.
The prover $\mathscr{P}$ wants to prove to a verifier $\mathscr{V}$ that it knows the discrete logarithm of $\kb$, i.e. it knows $x\in \Int_q$ such that $x\basePoint=\kb$.\\
So the common inputs are $\G,\basePoint$ and $\kb$, while the secret input of $\mathscr{P}$ is $x$.\\
The protocol works as follows:
\begin{enumerate}
	\item $\mathscr{P}$ picks $r\in\Int_q$ uniformly at random in and computes $\ub=r\basePoint$.
	Then $\mathscr{P}$ sends $\ub$ to $\mathscr{V}$.
	\item $\mathscr{V}$ picks $c\in \Int_q$ uniformly at random and sends it to $\mathscr{P}$.
	\item $\mathscr{P}$ computes $z=r+cx$ and sends $z$ to $\mathscr{V}$.
	\item $\mathscr{V}$ computes $z\basePoint$.
	If $\mathscr{P}$ really knows $x$ it holds that $z\basePoint=\ub + c\kb$. If the equality does not hold, the verifier rejects.
\end{enumerate}
A detailed proof about the security of the algorithm can be found in \cite{schnorr1989efficient}.

\subsubsection{Schnorr Protocol Simulation}
We need to simulate the Schnorr protocol in two different ways: first we need to use it to extract the adversary's secret value, then we need to simulate it without knowing our secret value, tricking the opponent.

We can use the Schnorr protocol to extract the value $x$ from the adversary in this way:
    \begin{enumerate}
        \item Follow the standard protocol until the third point, obtaining $z$.
        \item Rewind the adversary to the second point and pick $\tilde{c}\ne c$.
        \item Follow the remaining part of the protocol, obtaining $\tilde{z}$.
        \item We can compute $\frac{z-\tilde{z}}{c-\tilde{c}}=\frac{(c-\tilde{c})x}{c-\tilde{c}}=x$.
    \end{enumerate}
    \begin{proof}[Sketch]
    Since the only extra hypothesis for $\tilde{c}$ is that $\tilde{c}\ne c$ we can suppose that $\tilde{c}$ has uniform distribution as well. Moreover $z$, once the verifier sent $c$ the value of $z$ is fixed, so the rewinding technique does not cause any problem.
    \end{proof}
    
    At the same time we need to be able to simulate the protocol without knowing~$x$.
    The simulation works as follows:
    \begin{enumerate}
        \item Follow the protocol until the second point, obtaining $c$.
        \item Rewind the adversary to the first point.
        The simulator picks $r$ randomly and computes $\tilde{\ub}=(-xc+r)\basePoint= -c(x\basePoint)+r\basePoint$.
        Under the discrete logarithm assumption and since $r,c$ are random element, this is indistinguishable from~$r\basePoint$.
        \item The simulator sends $\tilde{\ub}$ and the adversary sends $c$ again.
        \item The simulator sends $z=r-cx+cx = r$.
        \item The adversary checks that $z\basePoint=r\basePoint=\tilde{\ub}+c(x\basePoint)=-x c\basePoint+r\basePoint+ xc\basePoint.$
    \end{enumerate}
    \begin{proof}[Sketch]
    The tricky point of the simulation is the third point, when we need that the adversary sends the same $c$ it has previously sent, since sending a different $r$ could change the random choice of $c$ . This could be achieved introducing an equivocable commitment scheme, in this way we need only to change the decommitment value after receiving the adversary commitment.
    \end{proof}

\subsubsection{Equality of Discrete Logarithms}
This simple variant of the protocol allows to prove that two public elements are linked to the same secret value.

More formally, let $\G$ be a cyclic group of prime order $q$, let $\basePoint, \bar{\basePoint}$ be generators of $\G$, and finally let $\kb, \bar{\kb} \in \G$, $x \in \Int_q$.
The prover knows $x$ and wants to convince the verifier that:
\begin{equation}\label{zkpdlogeq}
    x\basePoint = \kb \quad \text{and} \quad x\bar{\basePoint} = \bar{\kb},
\end{equation}
without disclosing $x$. The values of $\basePoint$, $\kb$, $\bar{\basePoint}$ and $\bar{\kb}$ are publicly known.
\\
\\
The protocol proceeds as follows:
\begin{enumerate}
    \item The prover generates a random $r$ and computes  $\ub = r\basePoint \quad \text{and} \quad \bar{\ub} = r\bar{\basePoint}$, then sends $(\ub, \bar{\ub})$ to the verifier.
    
    \item The verifier computes a random $c\in \{0, 1\}$ and sends it to the prover.
    
    \item The prover creates a response $s = r + c \cdot x$ and sends $s$ to the verifier.
    
    \item The verifier checks that $s\basePoint = c \kb + \ub$, $s\bar{\basePoint} = c\bar{\kb} + \bar{\ub}$.
    If the check fails the proof fails and the protocols aborts.
    
    \item The previous steps are repeated $\tau = \mathrm{poly}(\log_2(q))$ times, i.e. the number of repetitions is polynomial in the length of $q$ (the security parameter).
\end{enumerate}
A detailed analysis of the protocol and its security can be found in~\cite{spadafora2020coercion}.
\end{appendices}
\end{document}